\documentclass[reprint,amsmath,amssymb,aps,pra,floatfix]{revtex4-1}
\usepackage{graphicx}
\usepackage{dcolumn}
\usepackage{bm}
\usepackage[T1]{fontenc}
\usepackage{subfig}
\usepackage{physics}
\usepackage{tikz}
\usepackage{amsmath}
\usepackage{amsthm}
\usepackage{mathtools}
\usepackage[normalem]{ulem} 
\usepackage[colorinlistoftodos]{todonotes}
\newtheorem{theorem}{Theorem}
\newtheorem{prop}[theorem]{Proposition}
\newcommand{\Id}{{\rm 1\hspace{-0.9mm}l}}
\newcommand{\C}{\ensuremath{\mathbb{C}}}
\renewcommand{\AA}{\ensuremath{\mathcal{A}}}
\newcommand{\BB}{\ensuremath{\mathcal{B}}}
\newcommand{\proj}[1]{\ensuremath{\ketbra{#1}{#1}}}
\newcommand{\p}{\ensuremath{p_{\mathrm{I}}}}
\newcommand{\pp}{\ensuremath{p_{\mathrm{II}}}}
\newcommand{\defeq}{\mathrel{\vcenter{\baselineskip0.5ex \lineskiplimit0pt
                     \hbox{\scriptsize.}\hbox{\scriptsize.}}}%
                     =}
\begin{document}

\title{Local certification of unitary operations}
\author{Ryszard Kukulski}
\affiliation{Faculty of Physics, Astronomy and Applied Computer Science,\\
ul. Łojasiewicza 11, Jagiellonian University, 30-348 Kraków, Poland}

\author{Mateusz Stępniak}
\affiliation{Quantumz.io Sp. z o.o., Puławska 12/3, 02-566 Warsaw, Poland}

\author{Kamil Hendzel}
\affiliation{Quantumz.io Sp. z o.o., Puławska 12/3, 02-566 Warsaw, Poland}

\author{Łukasz Pawela}
\email{lpawela@iitis.pl}
\affiliation{Institute of Theoretical and Applied Informatics, Polish Academy of Sciences, Ba{\l}tycka 5, 44-100 Gliwice, Poland}

\author{Bartłomiej Gardas}
\affiliation{Institute of Theoretical and Applied Informatics, Polish Academy of Sciences, Ba{\l}tycka 5, 44-100 Gliwice, Poland}

\author{Zbigniew Puchała}
\affiliation{Institute of Theoretical and Applied Informatics, Polish Academy of Sciences, Ba{\l}tycka 5, 44-100 Gliwice, Poland}

\date{\today}

\begin{abstract}
    In this work, we analyze the local certification of unitary quantum
    channels, which is a natural extension of quantum hypothesis testing. A
    particular case of a quantum channel operating on two systems corresponding
    to product states at the input, is considered. The goal is to minimize the
    probability of the type II error, given a specified maximum probability of
    the type I error, considering assistance through entanglement with auxiliary
    systems. Our result indicates connection of the local certification problem
    with a product numerical range of unitary matrices. We show that the optimal
    local strategy does not need usage of auxiliary systems and requires only
    single round of one-way classical communication. Moreover, we compare local
    and global certification strategies and show that typically local strategies
    are optimal, yet in some extremal cases, where global strategies make no
    errors, local ones may fail miserably. Finally, some application for local
    certification of von Neumann measurements are discussed as well.
\end{abstract}

\maketitle

\section{Introduction}

In quantum information theory, a well-known problem is the discrimination of
states and quantum channels, as solved by Helstrom~\cite{helstrom}. This problem
involves distinguishing which state or channel from a given pair we are dealing
with, based on a prepared measurement. It plays a pivotal role in the
comprehension and manipulation of quantum systems. The ensuing step is
certification, a process aimed at confirming whether a given hypothesis
regarding the state, channel, or measurement holds true; this is achieved by
contrasting it with an alternative hypothesis. Certification safeguards the
integrity and reliability of quantum operations, rendering it indispensable for
quantum computing and communications. The mathematical explanation of the 
\emph{modus operandi} of quantum computers
involves the use of quantum operations and channels. This paper focuses on 
\textbf{the
local certification of unitary operations}, a
technique which is crucial for enhancing quantum computing applications, and
developing quantum algorithms and error correction strategies. This
certification is useful for benchmarking quantum devices, thus steering the
progression of quantum algorithm design.

Great deal of work has been done in the domain of local 
certification and distinguishability of quantum states. The research were 
focused on presenting conditions for a finite set of orthogonal
quantum states to be distinguishable by local 
operations~\cite{ghosh2004distinguishability, nathanson2005distinguishing, 
duan2009distinguishability, childs2013framework, zhang2017local}. It has 
been noticed, that local certification strategies not always are as powerful 
as global strategies involving usage of a quantum entanglement. Many 
examples of sets of orthogonal quantum states that cannot be perfectly 
certified were discovered in the literature~\cite{fan2004distinguishability, 
zhang2016local}; also in the domain of mixed quantum states 
\cite{calsamiglia2010local}.

The following up research about local discrimination of unitary channels has 
been built upon the results concerning quantum states. It was observed that 
there exist discrimination problems for which local procedures may be 
optimal as well as there exist problems for which they perform 
poorly~\cite{matthews2010entanglement}. Some conditions concerning optimal 
local discrimination strategies were already proposed 
too~\cite{bae2015discrimination}. Eventually, couple of works focused on 
many copies scenario, were it was shown that any two different unitary 
operations acting on an arbitrary multipartite quantum
system can be perfectly distinguishable by local operations and classical 
communication when a finite number of runs~\cite{duan2007entanglement, 
duan2008local, cao2016minimal}. However, in the literature, the problem of 
certification of unitary channels has not been yet considered. Taking up 
this challenge, in this work we explore a scenario where two parties, having 
access to a shared quantum unitary channel, engage in its certification. We 
compare local certification strategies with global ones and find the optimal 
and resource efficient local certification strategies.

Certification is closely related to statistical hypothesis testing, which is a
fundamental concept in statistical decision
theory~\cite{understanding_hypotesis_testing}. We consider a system with two
hypotheses: the null hypothesis ($H_0$) and the alternative hypothesis ($H_1$).
The null hypothesis intuitively corresponds to a promise about the system given
by its creator. By performing a test, we decide which hypothesis to accept as
true. A type I error occurs if we reject the null hypothesis when it is true.
The probability of this error occurring is called the level of significance. On
the other hand, a type II error occurs when we accept the null hypothesis, even
though it is false. We want to minimize the type II error given an assumed level
of significance. This approach is commonly referred to as certification. It
turns out that the concept of the numerical range of a matrix is an useful tool
in such issues. The numerical range provides
insights into the spectral and structural properties of matrices, making it important in quantum 
mechanics~\cite{product_usefull}.

\section{Preliminaries}
In this section, we elucidate the foundational principles underlying the
certification processes under consideration. 

\subsection{Notation}
In this work, we will encounter the notation of quantum states, quantum measurements and quantum 
channels. In order to set notation, we mention them here briefly (see \cite{watrous} for 
more detailed explanation).

We say that operator $\rho$ represents a quantum state defined on a system of dimension $d$ if 
this operator is positive semidefinite ($\rho \ge 0$) with unit trace ($\tr \rho = 1$). An 
linear map $\Psi$ will be called quantum channel if it is completely positive and trace 
preserving map (CPTP). We will consider a special family of quantum channels known as unitary 
channels. For an unitary matrix $U$ we define unitary quantum channel $\Psi_U$ by
\begin{equation}
    \Psi_U(\rho)= U \rho U^\dagger,
\end{equation}
where $\rho$ is an input quantum state $\rho$.

In a finite-dimensional case, quantum measurement, also called as positive operator-valued 
measure (POVM), is be represented by a set of positive semidefinite operators $\Omega = 
\{\Omega_i\}_i$ (also known as effects), such that $\sum_i \Omega_i = \Id$. 
According to the Born rule, for a given 
state $\rho$ the probability of obtaining the measurement outcome $i$ is given by $p_i = \Tr 
\rho \Omega_i$. The special subclass of quantum measurements consists of Von 
Neumann measurements. They fulfill the additional requirement that all 
effects $\Omega_i$ are rank-one projectors. Hence, for a von Neumann 
measurement acting on a state of dimension $d$ there are exactly $d$ effects 
$\Omega_i$ which are pairwise
orthogonal. This simple observation allows us to parameterize a 
$d$-dimensional
von Neumann measurement using a unitary matrix $U$, $P_U = 
\{U\proj{i}U^\dagger\}_{i=1}^d$. As a shorthand notation, we will write 
$\ket{u_i} \defeq U\ket{i}$. We can
associate a measure-and-prepare channel with a von Neumann measurement
\begin{equation}
P_U(\rho) = \sum_{i=1}^{d}  \bra{u_i} \rho \ket{u_i} \ketbra{i}{i}.
\end{equation}

Finally, in this work we will consider the family of bipartite quantum channels and measurements 
that can be realized by using only local operations and classical communication (LOCC) between 
two involved parties, say Alice and Bob. Let us assume that Alice has access to a system 
$\AA$ and Bob to a system $\BB$. Then, we say that bipartite quantum channel 
$\Psi$ (or measurement $\Omega$) is LOCC with respect to the partition 
$\AA:\BB$, if it can be realized by using local quantum 
operations on $\AA$ and $\BB$ separately, and by sharing classical information 
between $\AA$ and $\BB$. Such operations will be used notoriously in this work 
to prepare input quantum states and POVMs without creating quantum entanglement between involved 
parties. 

\subsection{Numerical ranges of a matrix}
In the context of certifying unitary channels, the numerical range and
the product numerical range play a pivotal role. The \emph{numerical range} of a
square matrix $X$ of size $d$, is defined as a subset of the complex plane:
\begin{equation}
    W(X)= \{ \bra{\psi} X \ket{\psi} : \braket{\psi}{\psi} = 1, \ket{\psi} \in \C^d \}.
\end{equation}
The set $W(X)$ is compact and convex; an in depth discussion of its properties and application 
can be found 
in~\cite{numerical_range, numerical_shadow}.

The \emph{product numerical range} of a square matrix $X$ of size $d_1 \cdot 
d_2$ with the 
partition $d_1: d_2$ is defined as:
\begin{equation}
\begin{split}
    W^{\otimes}_{d_1:d_2}(X) &= \{ (\bra{\psi_1} \otimes \bra{\psi_2}) X (\ket{\psi_1} \otimes 
    \ket{\psi_2}) : \\
    &\phantom{{}={}} \braket{\psi_1}{\psi_1} = 1, \braket{\psi_2}{\psi_2} = 1, \\
     &\phantom{{}={}} \ket{\psi_1} \in \C^{d_1} ,  \ket{\psi_2} \in \C^{d_2} \}.
\end{split}
\end{equation}
The core properties of this object are described in~\cite{product}.

\subsection{Operational scenario}
\label{operational_scenario} 

\begin{figure}[htbp]
\includegraphics{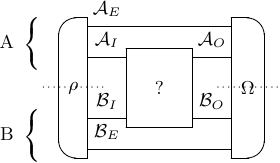}
\caption{ Schematic representation of the operational scenario for
unitary channel certification. Alice ($\mathrm{A}$) and Bob ($\mathrm{B}$) 
can prepare the initial state $\rho$ and the final measurement $\Omega$ by 
using LOCC operations to certificate if $? = \Psi_U$ or $? = \Psi_V$.  \label{schematic}}
\end{figure}

In the certification of a unitary channel, two parties, Alice and Bob, have an access to a 
quantum unitary channel. They have knowledge that the given channel is one of two possible 
unitary channels: $\Psi_U$, which will be identified with $H_0$ hypothesis or $\Psi_V$, which 
will be identified with $H_1$ hypothesis. Alice and Bob do not know which one is provided to 
them. Their goal is to find the best strategy to certificate that the given unknown channel is 
$\Psi_U$ and detect whenever the unknown operation is $\Psi_V$. More precisely, given the level 
of significance $\delta \ge 0$ as a parameter, they want to minimize the probability of making 
type II error provided the probability of making type I error is not greater than $\delta$.

The unknown operation has two inputs on the spaces $\AA_I$ and $\BB_I$, and two 
outputs on the spaces $\AA_O$ and $\BB_O$ (see Fig.~\ref{schematic}). Alice has 
an access to spaces 
$\AA_I$ and $\AA_O$ and auxiliary space 
$\AA_E$, while Bob to $\BB_I$ and $\BB_O$ and auxiliary system  
$\BB_E$, respectively. To produce an input state $\rho$ Alice and Bob are limited to use 
LOCC operations with respect to the partition $\AA_I \otimes \AA_E : \BB_I \otimes \BB_E$.
We will write in short that $\rho \in \mathtt{LOCC}$. Similarly, to produce a POVM $\Omega = 
\{\Omega_0, \Omega_1\}$ Alice and Bob are limited to use 
LOCC operations with respect to the partition $\AA_O \otimes \AA_E : \BB_O \otimes \BB_E$. We 
will use the notation $\Omega \in \mathtt{LOCC}$. Based 
on the results (classical labels) of their measurements, they decide, which unitary channel they 
are dealing with: label $0$ associated with the effect $\Omega_0$ indicates $H_0$, while label 
$1$ associated with the effect $\Omega_1$ indicates $H_1$.

\section{Main results}
\subsection{Unitary channel certification}\label{channelcertification}

In this section we will state the solution to the problem introduced in the 
Section~\ref{operational_scenario}. Let $\AA_I = \AA_O = \C^{d_1}$, $\BB_I = \BB_O = \C^{d_2}$. 
We are given two unitary matrices $U$ and $V$ of size $d_1 \cdot 
d_2$ and consider two hypotheses: 
\begin{itemize}
    \item \(H_0\): The operation is \(\Psi_U\).
    \item \(H_1\): The operation is \(\Psi_V\).
\end{itemize}

As mentioned previously (see also Fig.~\ref{schematic}), Alice and Bob can prepare the input 
state $\rho$ and the measurement $\Omega = \{\Omega_0, \Omega_1\}$ by using LOCC operations with 
respect to the partitions $\AA_I \otimes \AA_E : \BB_I \otimes \BB_E$ and $\AA_O \otimes \AA_E : 
\BB_O \otimes \BB_E$, respectively. The dimension of spaces $\AA_E$ and $\BB_E$ are arbitrary.
Alice and Bob accept accept the null hypothesis if the measurement result is $\Omega_0$, 
otherwise, they reject it. Thus, we arrive at the following formulas for the probabilities of
type I and type II errors:

\begin{equation}
    \begin{split}
    \p(\Omega, \rho) &= \tr\left(\Omega_1 (\Psi_U \otimes \Id_{\AA_E \otimes 
    \BB_E})(\rho)\right), \\
    \pp(\Omega, \rho)&=  \tr\left(\Omega_0 (\Psi_V \otimes \Id_{\AA_E \otimes 
    \BB_E})(\rho)\right).
\end{split}
\end{equation}
Certification requires minimizing $\pp$ under the condition $\p \leq \delta$, where $\delta$ is 
the desired significance level. It leads to the following optimization 
problem:
\begin{equation}\label{funkcja_celu}
    \pp(U, V) \defeq \min \{
    \pp(\Omega, \rho): \p(\Omega, \rho) \leq \delta, \rho, \Omega \in \mathtt{LOCC} \}.
\end{equation}

The remainder of this section is devoted to proving the
following theorem. This result shows the relation between two-party
certification and the product numerical range.

\begin{theorem}\label{thm:unitary_channels}
Consider the problem of two-point certification of unitary channels with
hypotheses
\begin{itemize}
    \item \(H_0\): The operation is \(\Psi_U\).
    \item \(H_1\): The operation is \(\Psi_V\).
\end{itemize}
defined as in Section~\ref{channelcertification} for  unitary 
matrices $U$ and $V$ of size $d_1 \cdot d_2$, and statistical significance 
$\delta \in [0,1]$. Let 
$z$ be the euclidean 
distance between $0$ and $W^{\otimes}_{d_1:d_2}(V^\dagger U)$, that is 
\begin{equation}\label{eq:z}
z \defeq \min\{ |x|: x \in  W^{\otimes}_{d_1:d_2}(V^\dagger U) \}.
\end{equation} Then, for the most 
powerful 
test utilizing LOCC operations, the
probability of the type II error yields
\begin{equation}\label{prob-local}
\begin{split}
\pp(U, V)= \begin{cases}
0, & z \le \sqrt{\delta}, \\
\left(z\sqrt{1-\delta} -  \sqrt{1 - 
z^2}\sqrt{\delta}\right)^2, & z > \sqrt{\delta}.
\end{cases}
\end{split}
\end{equation} 
\end{theorem}
\begin{proof}
By utilizing local quantum operations and classical communication, Alice and 
Bob can prepare any separable quantum state $\rho = \sum_j p_j \proj{a_j} 
\otimes \proj{b_j},$ where $\proj{a_j}$ are defined on $\AA_I \otimes \AA_E$ 
and $\proj{b_j}$ on $\BB_I \otimes \BB_E$. 

Let us assume that Alice and Bob take a product 
state, that is $\rho = \proj{a,b} = \proj{a} \otimes \proj{b}$. Then, 
depending which 
hypothesis is true, we obtain the one of the following pure states
\begin{equation}
\begin{split}
H_0: &\quad \proj{h_0} \defeq (\Psi_U \otimes \Id_{\AA_E \otimes 
    \BB_E})(\rho),\\
H_1: & \quad \proj{h_1} \defeq (\Psi_V \otimes \Id_{\AA_E \otimes 
    \BB_E})(\rho).
\end{split}
\end{equation}
We consider two cases. If $|\braket{h_0}{h_1}| \le \sqrt{\delta}$, 
then according to \cite[Theorem 1 \& Corollary 1]{optimal_certification} the 
best pure states certification strategy $\widetilde{\Omega}$ with 
significance level $\delta$ is of the form: $\widetilde{\Omega} = 
\{\widetilde{\Omega_0}, \widetilde{\Omega_1}\}$, where $\widetilde{\Omega_0} 
= \proj{\omega}$ for $\ket{\omega} = 
\frac{\ket{\widetilde{\omega}}}{\| \widetilde{\omega} \|}$ and 
$\ket{\widetilde{\omega}} = \ket{h_0} - \braket{h_1}{h_0} \ket{h_1}$. The 
operator $\proj{\omega}$ satisfies $|\braket{h_1}{\omega}|^2 = 0$ and 
$|\braket{h_0}{\omega}|^2 \ge 1 - \delta$. Therefore, the 
states $\proj{h_1}$ and $\proj{\omega}$ are orthogonal and we can find a 
LOCC measurement $\Omega = \{\Omega_0, \Omega_1\}$ \cite{walgate2000local}, 
such that $\tr(\Omega_0 
\proj{\omega}) = 1$ and $\tr(\Omega_1 \proj{h_1}) = 1$. Alice and Bob choose 
$\Omega$ as their measurement and achieve $\p(\Omega, \rho) = \tr(\Omega_1 
\proj{h_0}) \le 1 - \tr(\proj{\omega}\proj{h_0}) \le \delta$ and $ 
\pp(\Omega, \rho) = 1 - \tr(\Omega_1 \proj{h_1}) = 0$, which is the optimal 
solution in that case.

If $|\braket{h_0}{h_1}| > \sqrt{\delta}$, 
then according to \cite[Theorem 1 \& Corollary 1]{optimal_certification} the 
best pure states certification strategy $\widetilde{\Omega}$ with 
significance level $\delta$ is of the form: $\widetilde{\Omega} = 
\{\widetilde{\Omega_0}, \widetilde{\Omega_1}\}$, where $\widetilde{\Omega_0} 
= \proj{\omega}$ for $\ket{\omega} = 
\sqrt{1-\delta}\frac{\braket{h_0}{h_1}}{|\braket{h_0}{h_1}|} \ket{h_0} - 
\sqrt{\delta} \ket{h_0^\perp}$ and 
$\ket{h_0^\perp} = 
\ket{\widetilde{h_0^\perp}} / \|\widetilde{h_0^\perp}\|$, where 
$\ket{\widetilde{h_0^\perp}} = \ket{h_1} - \braket{h_0}{h_1} \ket{h_0}$. The 
operator $\proj{\omega}$ satisfies $|\braket{h_1}{\omega}|^2 = 
(\sqrt{1-\delta} |\braket{h_0}{h_1}| - \sqrt{\delta} \sqrt{1 - 
|\braket{h_0}{h_1}|^2})^2$ and 
$|\braket{h_0}{\omega}|^2 = 1 - \delta$. Here, for the orthogonal states 
$\proj{\omega}$ and $\proj{\omega^\perp}$, where $\ket{\omega^\perp} = 
\sqrt{\delta} \ket{h_0} + \sqrt{1-\delta} 
\frac{\braket{h_1}{h_0}}{|\braket{h_0}{h_1}|} \ket{h_0^\perp}$ we can find a 
LOCC measurement $\Omega = \{\Omega_0, \Omega_1\}$ \cite{walgate2000local}, 
such that $\tr(\Omega_0 
\proj{\omega}) = 1$ and $\tr(\Omega_1 \proj{\omega^\perp}) = 1$. Alice and 
Bob choose 
$\Omega$ as their measurement and achieve $\p(\Omega, \rho) = \tr(\Omega_1 
\proj{h_0}) \le 1 - \tr(\proj{\omega}\proj{h_0}) = \delta$. To calculate 
$\pp(\Omega, \rho)$ notice that $\ket{h_1}$ is spanned in the basis 
$\ket{\omega}, \ket{\omega^\perp}$, that is $\ket{h_1} = c_1 \ket{\omega} + 
c_2 \ket{\omega^\perp}$. Therefore, we get $\pp(\Omega, \rho) = \tr(\Omega_0 
\proj{h_1}) = |c_1|^2\tr(\Omega_0 
\proj{\omega}) = |c_1|^2$, where we used $\Omega_0\ket{\omega^\perp} = 
\ket{\omega^\perp} - \Omega_1\ket{\omega^\perp} = \ket{\omega^\perp} - 
\ket{\omega^\perp} = 0$. Eventually, $\pp(\Omega, \rho) =|c_1|^2 = 
|\braket{h_1}{\omega}|^2$, which is the optimal solution in that case.

We have showed that the optimal measurement $\Omega$ for product state 
$\rho$ gives 
\begin{equation}
\pp(\Omega, \rho) =
\begin{cases}
0, & |x| \le \sqrt{\delta}, \\
(\sqrt{1-\delta} |x| - \sqrt{\delta} \sqrt{1 - 
|x|^2})^2, & |x| > \sqrt{\delta},
\end{cases}
\end{equation}
where $x \defeq \braket{h_1}{h_0} = \bra{a, b} \left(V^\dagger U \otimes 
\Id_{\AA_E \otimes \BB_E}\right) \ket{a,b}.$ As the function $|x| \mapsto 
\pp(\Omega, \rho)$ is non-decreasing, Alice and Bob choose $\ket{a,b}$ that 
minimizes $|x|$. 

Observe, the choice of the optimal product input state $\rho = \proj{a,b}$ 
is independent of the significance level $\delta$. Therefore, no separable 
state $\rho$ will provide better result than product state. 

Finally, let us assume that $\rho = \proj{a_0, b_0}$ is the optimal product 
input 
state - it minimizes $\proj{a,b} \mapsto |\bra{a, b} \left(V^\dagger U 
\otimes 
\Id_{\AA_E \otimes \BB_E}\right) \ket{a,b}|.$ Let $\AA_E = \C^{e_1}$ and 
$\BB_E = \C^{e_2}$. Using the result from \cite[Lemma 2]{duan2008local} we get
\begin{equation}
\begin{split} 
&\bra{a_0, b_0} \left(V^\dagger U 
\otimes 
\Id_{\AA_E \otimes \BB_E}\right) \ket{a_0,b_0} \\
\in &
W^{\otimes}_{d_1e_1:d_2e_2}(V^\dagger U \otimes 
\Id_{\AA_E \otimes \BB_E})\\ =& 
W^{\otimes}_{d_1:d_2}(V^\dagger U).
\end{split}
\end{equation}
Hence, there is a quantum state $\proj{A_0}$ defined on the 
space $\AA_I$ and a quantum state 
$\proj{B_0}$ on the space $\BB_I$, such that 
$
\bra{a_0, b_0} \left(V^\dagger U 
\otimes 
\Id_{\AA_E \otimes \BB_E}\right) \ket{a_0,b_0} = \bra{A_0, B_0} V^\dagger U 
\ket{A_0,B_0}. 
$
Eventually, we observe that to 
minimize $|x|$ Alice and Bob do not need to use auxiliary systems $\AA_E$ 
and $\BB_E$ and the optimal state can be chosen as $\rho = \proj{A_0,B_0}$, 
which is defined on $\AA_I \otimes \BB_I$. Hence, for $z = \min \{ \left| 
\bra{a, b} V^\dagger U  \ket{a,b} \right|: 
\braket{a}{a} = \braket{b}{b} = 1,\ket{a} \in \C^{d_1}, \ket{b} \in 
\C^{d_2} \} = \min \{|x|: x \in  W^{\otimes}_{d_1:d_2}(V^\dagger U) \}$
 we achieve the desired result
\begin{equation}
\begin{split}
\pp(U, V)= \begin{cases}
0, & z \le \sqrt{\delta}, \\
(\sqrt{1-\delta} z - \sqrt{\delta} \sqrt{1 - 
z^2})^2, & z > \sqrt{\delta}.
\end{cases}
\end{split}
\end{equation} 
\end{proof}

\subsection{Optimal certification strategy}

The proof of Theorem~\ref{thm:unitary_channels} provides insight of the best 
certification strategy that Alice and Bob can utilize. Starting from the 
input state $\rho \in \mathtt{LOCC}$, we see that Alice does not have to 
create any entanglement between $\AA_I$ and $\AA_E$, what is more - she does 
not need to use auxiliary system $\AA_E$ at all. The same holds for Bob's 
systems. That optimal input state $\rho = \proj{a,b}$ defined on $\AA_I 
\otimes \BB_I$ is pure and product 
as well it minimizes $|\bra{a,b} V^\dagger U \ket{a,b}|$. It is independent 
from the significance level $\delta$. 

The explicit form of the optimal 
measurement $\Omega \in \mathtt{LOCC}$ is more complicated and relies 
heavily on the construction provided in \cite{walgate2000local} and the 
proof of Theorem~\ref{thm:unitary_channels}. Also, its description changes 
with $\delta$. Nevertheless, from the 
operational point of view, $\Omega$ can be simply 
realized by Alice and Bob. One party, let's say Alice, prepares an 
appropriate Von Neumann measurement $P_{R_A}$ on system $\AA_O$, where $R_A$ 
is an unitary rotation and sends 
the measurement result $i$ as a classical information to Bob. Then, he 
prepares a Von Neumann measurement $P_{R_B |_i}$ on $\BB_O$, conditioned on 
the information gained from Alice, where $R_B |_i$ is an unitary rotation. 
Bob's measurement result $j$ after classical post-processing $j \mapsto f(j) \in \{0,1\}$ 
indicates if the 
hypothesis $H_0$ should be accepted or rejected. We provide a schematic representation of the 
optimal strategy in Fig.~\ref{schematic2}.

\begin{figure}[htbp]
\includegraphics{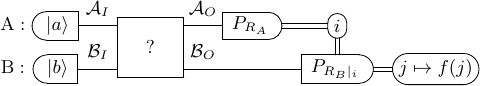}
\caption{ Schematic representation of the optimal operational scenario for
unitary channel certification. Alice ($\mathrm{A}$) prepares the initial 
pure state $\proj{a}$ and Bob ($\mathrm{B}$) prepares $\proj{b}$. The final 
measurement consists of Alice preparing $P_{R_A}$ and sending the result 
label $i$ to Bob, who prepares the measurement $P_{R_B|_i}$. The 
post-processed 
label $j \mapsto f(j) \in \{0,1\}$ of Bob's measurement certificate if $? = \Psi_U$ or $? = 
\Psi_V$.  
\label{schematic2}}
\end{figure}

\subsection{Local vs global certification of unitary 
channels}\label{sec:loc-glob}

In the case where a single party controls both inputs and outputs, the party 
can create entanglement between compound systems $\AA_I \otimes \AA_E$ and 
$\BB_I 
\otimes 
\BB_E$ (similarly between $\AA_O \otimes \AA_E$ and $\BB_O \otimes \BB_E$). 
In other words, the input state $\rho$ can be chosen arbitrarily, also the 
measurement $\Omega = \{\Omega_0, \Omega_1\}$. The
certification result $\pp^*(U, V)$ is expressed 
as~\cite{optimal_certification}:
\begin{equation}
\begin{split}
\pp^*(U, V)= \begin{cases}
0, & v \le \sqrt{\delta}, \\
\left(v\sqrt{1-\delta} -  \sqrt{1 - 
v^2}\sqrt{\delta}\right)^2, & v > \sqrt{\delta},
\end{cases}
\end{split}
\end{equation}
where 
\begin{equation}\label{eq:v}
v \defeq \min\{ |x|: x \in  W(V^\dagger U) \}.
\end{equation} 

As we can see, the certification results differ in local and global 
scenarios. Observe, that both results 
depend on the product $V^\dagger U$, hence, the results are unitarily 
invariant and we may assume that $V = \Id$ for the remainder of this 
section. In the global case we are interested in computing the distance $v( 
U) \defeq  \min\{ |x|: x \in  W(U) \}$, while for the local case we compute 
the 
distance $z_{d_1:d_2}(U) \defeq \min\{ |x|: x \in  W^{\otimes}_{d_1:d_2}(U) 
\}$. The forthcoming comparison of $z_{d_1:d_2}(U)$ and $v(U)$ will enable a 
comparative analysis between 
single-party and two-party certification scenarios. 

From the definition, for all unitary matrices $U$ of size $d_1 \cdot d_2$ it 
holds $v(U) \le z_{d_1:d_2}(U)$. The questions arise, can $v(U)$ be strictly 
lower than $z_{d_1:d_2}(U)$ and how much lower it can be? 

Let $d = d_1 = d_2$. Define 
\begin{equation}
U = \Id_{d^2} - \frac{2}{d} \proj{\Id_d},
\end{equation}
where $\ket{\Id_d} = \sum_{i=1}^d \ket{i,i}$. Observe that the eigenvalues of 
$U$ are $\pm 1$, so $v(U) = 0$ and therefore, $\pp^*(U, \Id_{d^2}) = 0$. On 
the other hand, for any normed vectors $\ket{\psi_1}, \ket{\psi_2} \in \C^d$ 
we have
\begin{equation}
\left| \bra{\psi_1, \psi_2} U \ket{\psi_1, \psi_2} \right| = \left| 1 - 
\frac{2}{d} |\braket{\psi_1}{\overline{\psi_2}}|^2 \right| \ge \frac{d-2}{d},
\end{equation}
where to saturate the last inequality we take $\ket{\psi_1} = \ket{\psi_2} = 
\ket{1}$. That means, $z_{d_1:d_2}(U) = (d-2)/d$ and when the local dimension 
goes to infinity, $d \to \infty$, then we have $\pp(U, \Id_{d^2}) \to 1 - 
\delta$.

We showed that in the extremal case, local and global strategies differ 
significantly. But what about a typical case? For the incoming analysis we 
assume that $U$ is Haar-random unitary 
matrix~\cite{zyczkowski1994random}. We have the following theorem.

\begin{theorem}\label{thm:random}
Let $U$ be a Haar-random unitary matrix of size $d_1 \cdot d_2$. For large 
enough product $d_1d_2$ we have
\begin{equation}
\mathbb{P}\left( z_{d_1:d_2}(U) = 0 \right) \ge 1 - \exp\left(-\frac{\log 
2}{2} \max(d_1^2, d_2^2)\right).
\end{equation}
\end{theorem}
\begin{proof}
Without loss of the generality let us assume that $d_2 \ge d_1$. For a fixed 
state $\ket{1} \in \C^{d_1}$ denote $M = (\bra{1} \otimes \Id) U (\ket{1} 
\otimes \Id)$. We have then $W\left(M\right) \subset 
W^{\otimes}_{d_1:d_2}(U)$. As $U$ is Haar distributed the matrix $M$ has the 
same distribution as $VM$, where $V$ is a Haar-random unitary matrix of size 
$d_2$ independent of $U$. The matrix $M$ is a truncation of $U$, hence, 
almost surely it has full rank ~\cite{zyczkowski2000truncations}. Continuing 
the 
reasoning, if $M = U_M Q_M$ is the polar decomposition of $M$, then $VM$ has 
the same distribution as $VQ_M$. Let $\lambda_i$ be eigenvalues of $V$ with 
corresponding eigenvectors $\ket{x_i}$. Then $\bra{x_i}VQ_M\ket{x_i} = 
\lambda_i \bra{x_i}Q_M\ket{x_i} \in 
W(VQ_M)$ for each $i$, where almost surely $\bra{x_i}Q_M\ket{x_i} > 0$. 
Therefore, if $0 
\in W(V)$, then there is a probability vector $(p_i)_i$ such that $\sum_i 
\lambda_i p_i = 0$. Let us define a probability vector $(q_i)_i$ given 
by $q_i = \frac{p_i}{\bra{x_i}Q_M\ket{x_i}} \left(\sum_j 
\frac{p_j}{\bra{x_j}Q_M\ket{x_j}} \right)^{-1}$. We obtain $\sum_i q_i 
\lambda_i \bra{x_i}Q_M\ket{x_i} = 0$ which implies $0 \in W(VQ_M)$. 
Combining all together we get
\begin{equation}
\begin{split}
\mathbb{P}\left( z_{d_1:d_2}(U) = 0 \right) &\ge \mathbb{P}\left( 0 \in W(M) 
\right) = \mathbb{P}\left( 0 \in W(VQ_M) \right) \\&\ge \mathbb{P}\left( 0 
\in W(V)\right) \ge 1 - \exp\left(-\frac{\log 2}{2} 
d_2^2\right),
\end{split}
\end{equation}
where the last inequality was proven in~\cite[Proposition 
19]{nechita2018almost} for large enough $d_2$.
\end{proof}

According to Theorem~\ref{thm:random}, when we are dealing with 
high-dimensional unitary channels $\Psi_U$ and 
$\Psi_V$, most of them can be perfectly certified. What is more, the optimal 
strategy is local, uses only once one-way classical communication channel
and 
does not need auxiliary systems (see Fig.~\ref{schematic2}). Such strategies 
are the most desirable in terms of used resources such as quantum 
entanglement~\cite{chitambar2019quantum}.

In this section, we learned that the gap between $v(U)$ and $z_{d_1:d_2}(U)$ 
may be huge in extremal cases. Also, we observed that typically in high 
dimensions both quantities are equal to zero. We provide more examples 
comparing local and global strategies in Appendix~\ref{appendix}.

\section{Application to Von Neumann measurement certification}

In the domain of von Neumann measurements certification, we are given two 
unitary matrices $U$ and $V$ of size $d_1 \cdot 
d_2$ and consider two hypotheses:
\begin{itemize}
    \item \(H_0\): The operation is \(P_U\).
    \item \(H_1\): The operation is \(P_V\).
\end{itemize}
The operational
paradigm is similar, yet it harbors a critical distinction: the output
transitions to a classical domain. Upon executing a joint measurement on their
respective quantum states, a classical label $i$ is generated. This label,
mutually acknowledged by Alice and Bob, serves as the foundational element for
subsequent certification processes. Thereafter, both parties then measure their
auxiliary systems, guided by the known label $ i$, to ascertain whether the
joint measurement was null of the alternative hypothesis. As in the unitary 
channel certification Alice and Bob strategy is similar (see 
Fig.~\ref{schematic}). They can prepare the input 
state $\rho \in \mathtt{LOCC}$ and the measurement $\Omega = \{\Omega_0, 
\Omega_1\} \in \mathtt{LOCC}$ while having access to auxiliary systems $\AA_E$ 
and $\BB_E$ of arbitrary dimension. Let $\widetilde{\pp}(U, V)$ indicates 
minimized 
probability of type II error under the condition that $\delta$ is given 
significance level for the introduced certification scheme of Von Neumann 
measurements. Then, we summarize our findings with the following proposition:

\begin{prop}\label{thm:measurements}

Consider the problem of two-point certification of Von Neumann measurements 
defined for unitary matrices $U$ and $V$ of size $d_1 \cdot 
d_2$, and statistical significance 
$\delta \in [0,1]$ with hypotheses 
\begin{itemize}
    \item \(H_0\): The operation is \(P_U\).
    \item \(H_1\): The operation is \(P_V\).
\end{itemize}
The most powerful test utilizing LOCC operations provides 
\begin{equation}
\begin{split}
\widetilde{\pp}(U, V) \ge \max\{& \pp(UE, VF): E, F \mbox{ is unitary }\\
				&\mbox{and 
diagonal}\}.
\end{split}
\end{equation} 
\end{prop}
\begin{proof}
Observe that action of each von Neumann measurement $P_U$, can
be expressed as $P_U = \Delta \Psi_{{(UE)}^{\dag}}$. Here, $\Delta$
is the completely dephasing channel $\Delta(X) = \sum_i \bra{i} X \ket{i} 
\proj{i}$ and $E$ is a diagonal unitary matrix. The channel $\Delta$ acting on 
$\AA_O \otimes \BB_O$ is equivalent to $\Delta = \Delta_A \otimes \Delta_B$, 
where $\Delta_A, \Delta_B$ are completely dephasing channels acting on $\AA_O$ 
and $\BB_O$, respectively. Let us fix $\rho_*, \Omega_* \in 
\mathtt{LOCC}$ as the optimal certification strategy achieving 
$\widetilde{\pp}(U, V)$. For the unitary channel certification between 
$\Psi_{{(UE)}^{\dag}}$ and $\Psi_{{(VF)}^{\dag}}$, where $E, F$ are diagonal, 
unitary matrices, we have 
\begin{equation}
\pp(\Omega_* (\Delta_A \otimes \Delta_B 
\otimes \Id_{\AA_E 
\otimes \BB_E}), \rho_*) = \widetilde{\pp}(P_U, P_V).
\end{equation}
Eventually, we get $\pp(UE, VF) \le \widetilde{\pp}(P_U, P_V)$, which ends the 
proof.
\end{proof}

\section{Summary}

In this study, we explored a scenario where two parties, having access to a 
shared quantum unitary channel, engage in its certification. Each party 
conducts individual measurements on their respective systems following 
channel utilization. We demonstrated in Theorem~\ref{thm:unitary_channels} 
that the certification challenge can be 
effectively transformed into an optimization problem involving the product 
numerical range. Following the proof of Theorem~\ref{thm:unitary_channels}, 
we concluded that the optimal local strategy does not need usage of 
auxiliary systems and parties involved need to utilize one-way classical 
communication channel. We provided the original considered scheme in 
Fig.~\ref{schematic} and the optimal and resource efficient scheme in 
Fig.~\ref{schematic2}.

In Section~\ref{sec:loc-glob} we compared local certification strategies 
with global ones. We observed that in the extremal case a global (single 
party) strategy can make no type II errors, while for the best local 
strategy the probability of making type II error approaches $1 - \delta$, 
where $\delta$ is the significance level. However, in 
Theorem~\ref{thm:random} we proved that typically, for high-dimensional 
unitary channels, local certification strategies are optimal and what is 
more, they make no type II errors. Assuming Haar distribution of unitary 
channels of size $d$, the probability of perfect local certification is no 
smaller 
than $1 - \exp(- \frac{\log 2}{2} d)$, which is approaching exponentially 
fast to $1$ as 
$d \to \infty$.

Regarding von Neumann measurements, our findings in 
Proposition~\ref{thm:measurements}
provide insights into the lower bound of the type II error, thus 
contributing partial but significant knowledge to the field of quantum 
measurement certification.

Our work opens new paths for future research. They could include more 
advanced comparison of $v(U)$ and $z_{d_1:d_2}(U)$ as well as finding 
effective ways of computing $z_{d_1:d_2}(U)$ for arbitrary $U$. From the 
operational point of view, the scenario where many players are involved in 
the local certification process seems to be interesting to explore. As the 
results of our work are based strongly on LOCC measurements' construction 
provided in \cite{walgate2000local} which is also valid in many players 
scenario, quick analysis suggests that the Eq.~\eqref{prob-local} will be 
valid in that case, but with replaced 
\begin{equation}
z \leftarrow \min\{|x|: x\ \in W_{d_1:\ldots:d_N}^\otimes(V^\dagger U)\},
\end{equation}
where $W_{d_1:\ldots:d_N}^\otimes$ is the product numerical range defined 
for $N$ parties~\cite{duan2008local}.

\section*{Acknowledgements}
MS and KH acknowledge support received from The National Centre for Research and
Development (NCBR), Poland, under Project No. POIR.01.01.01-00-0061/22, titled
\textit{VeloxQ: Utilizing Dynamic Systems in Decision-Making Processes Based on
Knowledge Acquired Through Machine Learning, Across Various Levels of
Complexity, in Industrial Process Optimization}, that aims to develop digital
solutions for solving combinatorial optimization problems. The focus is on
leveraging quantum-inspired algorithms and artificial intelligence to enhance
decision-making processes in industrial optimization contexts. The insights and
knowledge gained during the research phase of this project have significantly
shaped the content and findings of this article.

RK acknowledges financial support by the National Science Centre, Poland, 
under the contract 
number 2021/03/Y/ST2/00193 within the QuantERA II Programme that has 
received funding from the European Union’s Horizon 2020 research and 
innovation programme under Grant Agreement No 101017733.

This project was supported by the National Science Center (NCN), Poland, under
Projects: SONATA BIS 10, No. 2020/38/E/ST3/00269 (LP, BG) and OPUS, No
2022/47/B/ST6/02380 (ZP).

\bibliography{certification}

\appendix
\section{Examples}\label{appendix}

\subsection{Product unitary matrices}
When the unitary matrix is a product $U = U_1 \otimes U_2$ ($U_1$ is of size 
$d_1$ and $U_2$ is of size $d_2$) we can express the minimized distance 
$z_{d_1:d_2}(U)$ as
\begin{equation}
\begin{split}
z_{d_1:d_2}(U_1 \otimes U_2) &= \min\{ |x|: x \in  W^{\otimes}_{d_1:d_2}(U_1 
\otimes U_2) 
\}\\
&= \min\{ |xy|: x \in  W(U_1), y \in  W(U_2)  \}\\
&= v(U_1) v(U_2).
\end{split}
\end{equation}
If $0 \in W(U_1)$ (or $0 \in W(U_2)$), then the best local and global 
strategies make no type II errors. Assume now, that $0 \not\in  W(U_1)$ and 
$0 
\not\in  W(U_2)$. Let $1, \ldots, \alpha$ be the eigenvalues of $U_1$ 
ordered counterclockwise by their phases and let $1, \ldots, \beta$ be the 
eigenvalues of $U_2$ 
ordered counterclockwise by their phases. As zero is not in the numerical 
ranges of both 
matrices, we may write $z_{d_1:d_2}(U_1 \otimes U_2) = 
|1+\alpha||1+\beta|/4$.
To achieve the greatest difference between local and global strategies we 
take 
$\alpha = \beta = i$. Then, $1$ and $i \times i = -1$ are the eigenvalues of 
$U_1 \otimes U_2$, so $v(U_1 \otimes U_2) = 0$, but $z_{d_1:d_2}(U_1 \otimes 
U_2) = |1+i|^2/4 = 1/2$. 

On the other, if $U_1 = \Id$, then from the property of the numerical range 
$v(\Id \otimes U_2) = v(U_2) = z_{d_1:d_2}(\Id \otimes U_2)$. The value 
$v(U_2)$ can vary arbitrarily in the range $[0,1]$.

\subsection{Diagonal unitary matrices}
For a diagonal unitary matrix $U$ let us indicate its diagonal elements as 
$D_{i,j}$ for $i=1,\ldots,d_1$ and $j=1,\ldots,d_2$. We write
$\left| \bra{\psi_1, \psi_2} U \ket{\psi_1, \psi_2} \right| = \left| 
\sum_{i,j} p_i q_j D_{i,j}
 \right|$, where $(p_i)_i$ and $(q_j)_j$ are probability vectors that arise 
 from 
diagonals of $\proj{\psi_1}$ and $\proj{\psi_2}$, respectively. If all 
eigenvalues of $U$ are contained in an arch of length $\pi$, then one can 
show that to compute $z_{d_1:d_2}(U)$ it is sufficient to 
consider quadruples of the eigenvalues of the form $D_{i,j}$, where $i=i_1, 
i_2$ and $j=j_1, j_2$. In that case we optimize $| pqD_{i_1, j_1} + 
(1-p)q D_{i_2, j_1} + 
p(1-q) D_{i_1, j_2} + (1-p)(1-q) D_{i_2, j_2}|$ over $p,q \in [0,1]$. To 
achieve the greatest difference between $z_{d_1:d_2}(U)$ and $v(U)$ we put 
$D_{1,1} = 1, D_{1,2} = D_{2,1} = i, D_{2,2} = -1$. This configuration is 
equivalent with product case $\{xy: x,y \in \{1,i\}\}$, which we have 
already solved, that is $v(U) = 0$ and $z_{d_1:d_2}(U) = 1/2$.

\subsection{Upper-bound on $z$}
Following the property \cite[Property 9]{product}, the 
barycenter 
of the spectrum of $U$ always belongs to the product numerical range, 
$\tr(U) 
/ 
(d_1d_2) \in  W^{\otimes}_{d_1:d_2}(U)$. Hence, for any $U$ we have
\begin{equation}
z_{d_1:d_2}(U) \le \frac{|\tr(U)|}{d_1d_2}.
\end{equation}

\subsection{Numerical example}

Consider the family of matrices:
\begin{equation}
T_\alpha = (U \otimes V)^{(1 - \alpha)} (\Id_d - 2\ketbra{\psi}{\psi})^\alpha,
\label{eq:family}
\end{equation}
where $\ket{\psi} = \ket{\Id}/\sqrt{d}$ is the maximally entangled state of 
dimension $d$. and $U$,
$V$ are some unitary matrices. The family $T_\alpha$ allows us to smoothly
transition from a product unitary matrix to an arbitrary one. A Monte-Carlo
example of the product numerical for some values of parameter $\alpha$ is show
in Fig.~\ref{fig:product-shadow}. Aside from showing the general shape of the
numerical range the plots also show the density of points for each case.

\begin{figure*}[!h]
\centering
\includegraphics[width=0.45\textwidth]{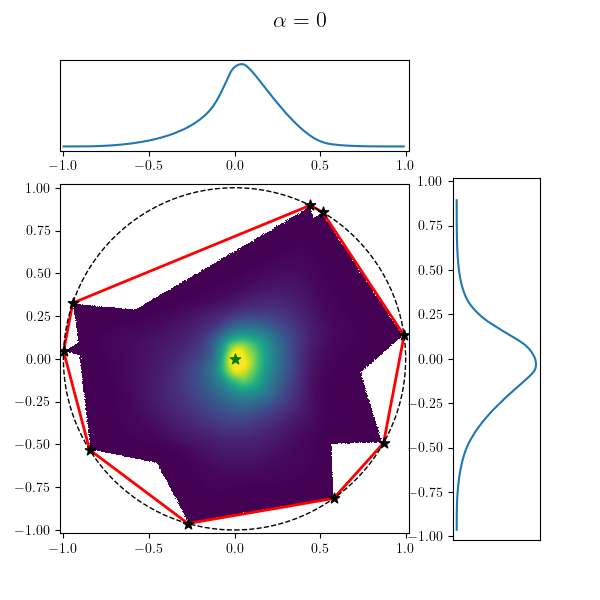}
\includegraphics[width=0.45\textwidth]{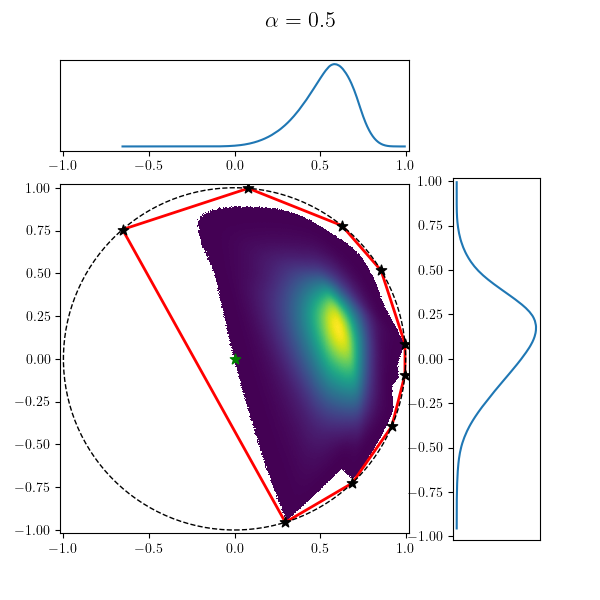}\\
\includegraphics[width=0.45\textwidth]{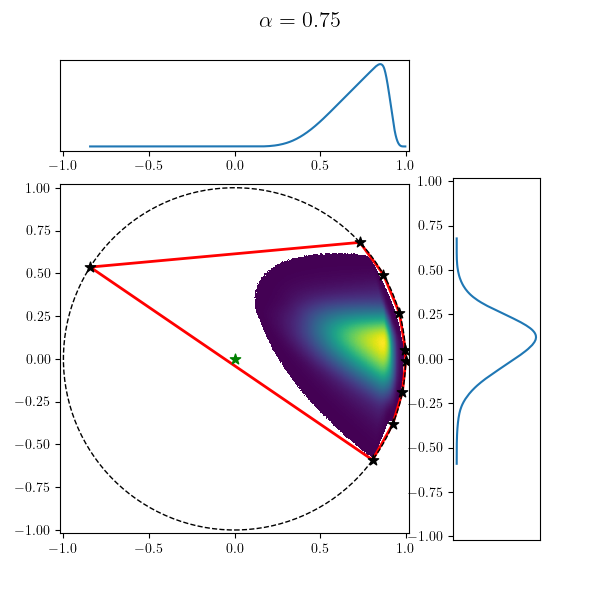}
\includegraphics[width=0.45\textwidth]{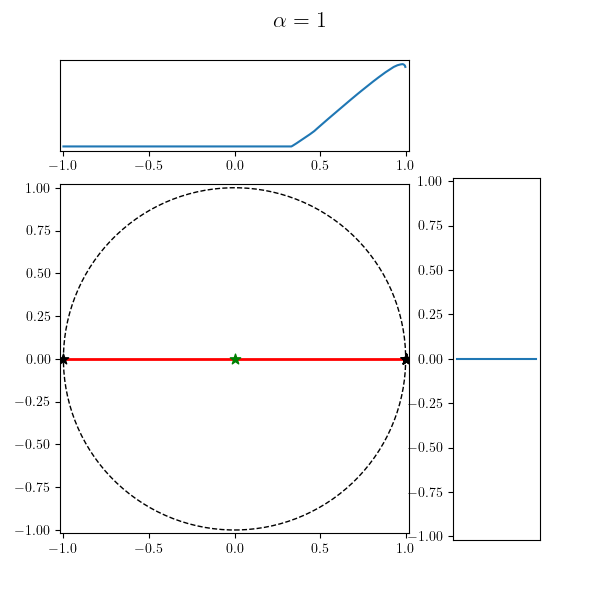}
\caption{Product numerical ranges of matrices from Eq.~\eqref{eq:family} for
$\alpha=0, 0.5, 0.75, 1$. The red polygon is the numerical range of $T_\alpha$,
the black stars are eigenvalues of $T_\alpha$. The product numerical range is
shown as an empirical histogram resulting from Monte Carlo simulation. The side
plots are marginal distributions.}\label{fig:product-shadow}
\end{figure*}
\end{document}